\documentclass[acmsmall, screen, nonacm]{acmart}

\AtBeginDocument{%
  }

\setcopyright{none}
\acmDOI{}
\acmISBN{}

\usepackage{cleveref}
\usepackage{pgfplots}
\usepackage{placeins}

\crefname{figure}{Fig.}{Figs.}
\Crefname{figure}{Fig.}{Figs.}

\pgfplotsset{compat=1.18}
\usetikzlibrary{arrows.meta,calc,positioning,shapes.geometric}
\graphicspath{{figures/}}

\definecolor{hlcpblue}{RGB}{0,92,153}
\definecolor{hlcporange}{RGB}{230,120,20}
\definecolor{hlcpgray}{RGB}{75,75,75}
\definecolor{hlcpgreen}{RGB}{40,130,85}

\begin{document}

\title[HLCP: State-Contingent Liquidity Activation]{HLCP: Secant--Tangent Geometry and Fully Backed
State-Contingent Liquidity Activation for Constant-Product AMMs}

\author{Hyoungsung Kim}
\email{hyoungsung@keti.re.kr}
\orcid{0000-0002-4307-2691}
\affiliation{%
  \institution{Korea Electronics Technology Institute}
  \country{Republic of Korea}
}

\author{Yong-Suk Park}
\authornote{Corresponding author.}
\email{yspark@keti.re.kr}
\orcid{0000-0002-6694-5125}
\affiliation{%
  \institution{Korea Electronics Technology Institute}
  \country{Republic of Korea}
}
\renewcommand{\shortauthors}{Kim and Park}

\begin{abstract}
Constant-product automated market makers (CPMMs) expose liquidity providers to
arbitrage when external market prices change before the pool price adjusts,
generating the adverse-selection loss measured by loss-versus-rebalancing (LVR).
Reducing the capital kept continuously on the trading curve can reduce LVR, but
also lowers liquidity depth and increases execution slippage. We examine whether
part of the capital can remain outside the trading curve with only limited
deterioration in swap execution. Our secant--tangent analysis relates active depth
to execution slippage and post-swap marginal price, showing diminishing execution
benefits from additional liquidity. Motivated by this diminishing benefit, we
study the Hybrid Liquidity--Collateral Pool (HLCP), a fully backed architecture
that initially scales both reserves of a full-depth CPMM by a fraction \(N\) and
holds the remaining paired tokens in a buffer. When a submitted order would cause
a large price movement, HLCP temporarily deploys part of this buffer in the pool's
current reserve ratio, expanding the curve homothetically before execution, and
recovers the activated share after settlement. To evaluate these execution effects
using historical market data, we first calibrate \(N\) in balanced synthetic CPMMs
from normalized 2025 WETH/USDC swaps observed in Uniswap V3. This calibration
identifies a share of capital that can remain in the buffer with a limited
execution gap to full depth. Holding this fraction fixed, a comparison using an
acute order constructed from historical data shows that buffer activation improves average execution and limits
price movement relative to a matched static CPMM and the passive V3 snapshot,
without adding capital.

\end{abstract}



\maketitle

\section{Introduction}
\label{sec:introduction}

Automated market makers (AMMs) provide liquidity for token swaps in
decentralized finance (DeFi) without relying on an order book.  In a
constant-product market maker (CPMM), the two reserve balances $x$ and $y$
follow the invariant $x\cdot y=k$, which defines the basic reserve curve
\cite{Uniswap_V2}.  In the fee-free model, a swap adds one asset to the pool
and removes the other while staying on this curve.  For a given starting
price and swap size, greater liquidity reduces execution slippage and limits
the change in the pool price.  Liquidity-provider (LP) capital on the curve
therefore determines the depth available to swaps.

This depth also leaves LP capital available to arbitrageurs when the external
market price moves first.  The pool's current price, determined by its reserve
ratio, can lag behind the external market price until arbitrageurs trade it
toward the new level, and
loss-versus-rebalancing (LVR)
measures the resulting LP loss relative to a frictionless rebalancing strategy
\cite{LVR,cartea2023predictable}.  Keeping more capital active improves swap
execution, but it also leaves a larger reserve base continuously exposed to
this repricing.  Holding capital outside the active curve has the opposite
effect: it reduces the reserve base that remains continuously exposed, but a
smaller active curve gives traders less output and moves the pool's
price farther from its starting level after the same swap.
The design problem is therefore to keep part of the capital outside the curve during
ordinary trading and restore it when a submitted swap needs more depth.

The cost of reducing active liquidity is greater execution slippage for the
same swap.  In the fee-free CPMM geometry,
this cost can be understood by comparing the trader's
average execution price with the pool's
price before the swap.  For a sale of $x$, the magnitude of the slope of the
line joining the pre-swap and post-swap reserve states gives the average
price received in $y$ per unit of $x$.  This line is the secant.  The tangent
at the starting state gives the initial pool price, and the tangent at
the final state gives the pool price left by the swap.  More active
liquidity brings both the average and final prices closer to the starting
price, with each additional unit of liquidity producing a smaller improvement
in average execution.  We derive these relations for both swap directions
and use the secant relation to measure the output lost when only a fraction
$N$ of a balanced CPMM remains active.  \Cref{sec:finite-order-geometry}
provides the exact derivation.

We use this geometry to study the Hybrid Liquidity--Collateral Pool (HLCP), a
fully backed architecture that divides an initially balanced reserve position
between an active CPMM and a paired-token buffer.  The fraction $N$ begins on
the active curve, while the remaining $1-N$ stays in the buffer.  At $N=1$,
the buffer is empty and HLCP behaves exactly like an ordinary CPMM.  At
$0<N<1$, part of both assets remains available for later use.  Before
executing a submitted swap, HLCP performs a fee-aware preview without changing
the pool state.  If the projected relative move in the pool price exceeds
the protocol threshold, HLCP transfers a bounded amount of both tokens from
the buffer to the active reserves in the same proportion as the current
reserves.  Adding both tokens in this proportion scales both reserve balances
by the same factor and expands the curve about the origin, a transformation
known as homothety.  In \cref{fig:homothety-overview}, the reserve state moves
outward along a straight line from the origin to the corresponding point on
the expanded curve.  The tangents at the two states remain parallel, so
liquidity increases while the current pool price stays unchanged.  The
submitted swap then executes exactly once on the enlarged curve.  After the
swap, HLCP withdraws the pro-rata share associated
with the activated reserves and returns it to the buffer in the token
composition produced by the swap.  Every unit of temporary depth is therefore
backed by tokens already held in the buffer.

\begin{figure}[htbp]
  \centering
  \resizebox{0.74\linewidth}{!}{
\begin{tikzpicture}[
  x=2.05cm,
  y=1.45cm,
  font=\footnotesize,
  line cap=round,
  line join=round
]
  \draw[-{Stealth[length=1.7mm]},semithick]
    (0,0) -- (4.65,0) node[below right=-1pt] {$x$};
  \draw[-{Stealth[length=1.7mm]},semithick]
    (0,0) -- (0,2.45) node[above left=-1pt] {$y$};
  \node[anchor=north east,inner sep=1pt] at (0,0) {$O$};

  \draw[hlcpgray,thick,densely dashed]
    plot[domain=0.44:4.40,samples=180,smooth] (\x,{1/\x});
  \draw[hlcpblue,very thick]
    plot[domain=0.99:4.40,samples=180,smooth] (\x,{2.25/\x});

  \draw[hlcpgray,densely dotted,semithick] (0,0) -- (2.13,1.3632);
  \draw[hlcpblue,very thick,-{Stealth[length=1.8mm]}]
    (1.28,0.8192) -- (1.845,1.1808);

  \draw[hlcporange,thick] (0.95,0.992) -- (1.55,0.608);
  \draw[hlcporange,thick] (1.575,1.392) -- (2.175,1.008);

  \node[circle,draw=hlcpgray,fill=white,thick,inner sep=1.8pt]
    at (1.25,0.80) {};
  \node[rectangle,draw=hlcpblue,fill=hlcpblue,inner sep=1.8pt]
    at (1.875,1.20) {};

  \node[align=center,inner sep=0pt] at (3.18,1.84)
    {$(x,y)\ \longmapsto\ \lambda(x,y)$,\quad $\lambda>1$\\[2pt]
     {\color{hlcporange}Parallel tangents: same marginal price}};

  \draw[hlcpgray,thick,densely dashed] (0.35,2.66) -- (0.65,2.66);
  \node[anchor=west,inner sep=0pt] at (0.75,2.66)
    {Original CPMM};
  \draw[hlcpblue,very thick] (2.65,2.66) -- (2.95,2.66);
  \node[anchor=west,inner sep=0pt] at (3.05,2.66)
    {Scaled CPMM};
\end{tikzpicture}}
  \caption{Homothetic expansion of a CPMM about the origin.
  The corresponding reserve states lie on the same ray and have parallel
  tangents, preserving the pool price.}
  \Description{An original dashed reserve curve and an expanded solid curve
  are shown with the origin visible. One ray from the origin joins two
  corresponding states, with an outward arrow showing expansion. Two short
  parallel tangents indicate the same marginal price at those states.}
  \label{fig:homothety-overview}
\end{figure}

Experiment~1 asks how much capital can remain in the buffer while limiting
the additional loss in swap output relative to full depth.  We sample 2025
Swap events from the WETH/USDC Uniswap V3 pool selected by annual Swap volume
\cite{uniswapv3core2021}.  These records supply observed order sizes, prices,
and pool capital for the comparison.  We normalize each order by the pool's
marked capital at its pre-swap price, so observations from different times
can be compared on a common scale.  For each eligible Swap, we compare an
$N$-scaled balanced synthetic CPMM with its full-depth counterpart.  We use
the selected pool's 5-bp swap fee as a reference scale for this additional
loss.  We calibrate $N$ using the 95th percentile (P95) of the output gap, so
that at least 95\% of the eligible sampled Swaps have an additional output
loss no greater than this threshold.  This annual single-pool calibration
gives the descriptive active fraction $N^*=0.815$, which supplies the fixed
reserve split used in the second experiment.

Experiment~2 holds this reserve split fixed and examines whether the paired
buffer can restore depth for one acute order.  A price move selected from
Coinbase data sets the shock magnitude, and a historical Uniswap V3 snapshot
determines the order needed to produce that move.  We apply the same order,
starting price, marked-capital scale, and fee to historical passive V3, a
matched static CPMM initialized at $N^*$, and HLCP.  The static CPMM starts
with the same active reserves and buffer as HLCP but keeps the buffer
dormant, allowing their comparison to isolate what activation adds.  V3
provides a comparison with an observed allocation of liquidity across price
ranges.  In this case, using part of the paired buffer improves the trader's
average execution price and leaves the pool's final price closer
to its starting level than the matched static CPMM.

Our contributions are as follows:
\begin{enumerate}
  \item \textbf{Swap execution geometry.}  We derive exact
  direction-specific relations between a swap's average execution price, its
  post-swap pool price, and the active liquidity on a constant-product
  curve.  The relations show the diminishing improvement in average execution
  and provide the equation used to calibrate $N$.

  \item \textbf{Fully backed buffer activation.}  We specify an
  order-conditioned rule that draws on paired reserves to expand the CPMM
  without changing its pre-swap pool price.  Bounded activation and
  pro-rata recovery preserve token backing throughout the transition.

  \item \textbf{Annual calibration and acute execution.}  We calibrate the
  descriptive value $N^*=0.815$ from normalized 2025 WETH/USDC Swap events and
  carry it unchanged into one common-swap comparison.  The matched static CPMM
  isolates the effect of activation, while the historical V3 state supplies
  the observed passive-liquidity profile for the contextual comparison.
\end{enumerate}

The remainder of this paper is structured as follows.  \Cref{sec:background}
defines the reserve state, price, fee, LVR, and related-work context.
\Cref{sec:architecture} develops the swap geometry and the complete HLCP state
transition.  \Cref{sec:method} presents the empirical design, while the annual
calibration and acute execution results appear in \cref{sec:evaluation}.
\Cref{sec:conclusion} summarizes the main findings.

\section{Background and Related Work}
\label{sec:background}

This section reviews how a CPMM uses token reserves to execute swaps and how
liquidity allocation affects the depth available to traders.
\Cref{sec:cpmm-background} introduces the reserve balances, swap prices, and
fee treatment.  \Cref{sec:allocation-background} compares concentrated
liquidity, dynamic fees, and just-in-time liquidity with HLCP's use of a paired
buffer.  \Cref{sec:economic-background} then relates the execution effects
studied in this paper to the LP losses measured by LVR.

\subsection{CPMM reserves and swap execution}
\label{sec:cpmm-background}

A CPMM holds reserves of two different tokens and executes a swap by receiving
one token and paying out the other.  Let $x$ and $y$ denote the amounts of the
first and second tokens held in the pool, respectively.  In the fee-free
model, the output is chosen so that the product of these reserves stays
constant.  We write this product as $L^2$, where $L$ describes the pool's
liquidity scale, and express prices as units of the second token per unit of
the first.  Before fees, the exchange rate approached by an infinitesimally
small sale of the first token is the pool's current marginal price,
$P_{\mathrm{amm}}$.  These quantities satisfy
\begin{equation}
  xy=L^2,
  \qquad
  P_{\mathrm{amm}}=\frac{y}{x},
  \qquad
  x=\frac{L}{\sqrt{P_{\mathrm{amm}}}},
  \qquad
  y=L\sqrt{P_{\mathrm{amm}}}.
  \label{eq:cpmm-state}
\end{equation}
\Cref{eq:cpmm-state} explains why the expansion in
\cref{fig:homothety-overview} preserves the pool price.  At a fixed
$P_{\mathrm{amm}}$, both reserves are proportional to $L$, so increasing $L$
scales $x$ and $y$ by the same factor.  Their ratio $y/x$, and hence the
marginal price, remains unchanged.  The parallel tangents at the
corresponding states in \cref{fig:homothety-overview} show this same price
before and after expansion.  HLCP uses this relationship to add paired
buffer tokens in the current reserve ratio.

A finite-size swap changes the reserve balances as it executes, so the
starting marginal price alone does not describe the trader's execution.
When the trader sells the first token, the amount of the second token
received divided by the amount of the first token supplied gives the average
execution price, $P_{\mathrm{eff}}$.  The resulting reserve pair determines
the marginal price after the swap, $P_{\mathrm{post}}$.  The pool's reserve
ratio need not match the price in an external market; we denote that external
valuation by $P_{\mathrm{ext}}$, expressed in the same units.
\Cref{sec:finite-order-geometry} derives the geometric relations between the
average execution price and the marginal prices before and after the swap,
illustrated in \cref{fig:geometry}.

The fee-free model describes how a swap moves along the reserve curve.
With swap fees, the input used to calculate the output differs from the
amount actually added to the reserves.
Let $f\in[0,1)$ be the input fee rate and write $\gamma=1-f$.  For a gross input
of $q$ units of the first token, the output calculation uses the effective
input $\gamma q$, while the full amount $q$ enters the reserves.  Writing
the second-token output as $\Delta y$ and the reserve balances after the
swap as $(x^+,y^+)$ gives
\begin{equation}
  \Delta y=y-\frac{xy}{x+\gamma q},
  \qquad
  x^+=x+q,
  \qquad
  y^+=y-\Delta y.
  \label{eq:sell-transition}
\end{equation}
For a swap in the other direction, the second token is supplied and the first
is received, with the same fee treatment and the two reserve roles exchanged.
The retained fee is therefore already included in the post-swap reserves.

\subsection{Liquidity allocation}
\label{sec:allocation-background}

The liquidity available to a swap also depends on how LP capital is allocated
across prices.  Uniswap V3 lets LPs choose the price intervals in which their
liquidity is active.  Only positions whose intervals contain the current
price contribute to active liquidity, which can change as the price crosses
a range boundary marked by an initialized tick \cite{uniswapv3core2021}.
Recent work analyzes how LPs select these ranges, the incentives behind their
choices, and empirical departures from equilibrium in concentrated-liquidity
markets \cite{tang2025game,llacertrotti2025jit}.  V3 therefore allocates
liquidity according to LP-selected price ranges.  HLCP controls how much of
its paired buffer is added to the active CPMM for a submitted swap.

Other approaches change the fee or add liquidity for an individual swap.
Dynamic fees adjust the fee charged for trading against the available
reserves \cite{lebedeva2025dynamicfee}.  Just-in-time liquidity temporarily
adds depth around a swap through an LP's decision based on information and
fee incentives \cite{llacertrotti2025jit}.  Uniswap V4 provides hooks that run
before and after a swap and can support such temporary liquidity provision
\cite{uniswapv4core2024}.  DualPool uses these hooks to deploy
concentrated-liquidity positions from inventory held outside the pool and
remove those positions after execution, with ordinary V4 liquidity expected
to be zero between swaps \cite{uniswapdualpooltech2026}.

Unlike DualPool, HLCP retains a continuously active CPMM between swaps.  It
uses a fixed, public rule to preview a submitted swap's price movement and
determine a bounded deposit from fully backed paired tokens already held in
its buffer.  The rule applies regardless of who submits the swap.  Depositing
both tokens in the current reserve ratio expands the active curve
homothetically, adding depth without changing the pre-swap pool price.  After
settlement, only the activated pro-rata share returns to the buffer in the
token composition left by the swap.

\subsection{Execution geometry and economic loss}
\label{sec:economic-background}

The execution cost of keeping capital outside the trading curve depends on
how active liquidity affects swap prices.  Prior work develops general
geometric descriptions of constant function market makers
\cite{angeris2024geometry} and relates liquidity and price sensitivity to
trading-function curvature \cite{angeris2022curvature}.  For the finite CPMM
swaps studied here, our analysis focuses on the relation between the trader's
average execution price and the pool's post-swap marginal price.

The same reserves that provide execution depth are also exposed to arbitrage
when external prices change before the pool's price adjusts.  Arbitrageurs
can then profit by trading against the stale pool price, imposing an
adverse-selection cost on LPs.  LVR measures this loss relative to a
frictionless strategy that rebalances at the external price
\cite{LVR,cartea2023predictable}.  Holding part of the capital in a paired
buffer reduces the amount kept continuously on the trading curve, but also
reduces the depth initially available to swaps.  The secant--tangent relations
provide a quantitative basis for calibrating HLCP's initial active fraction
and evaluating the execution effects of temporary buffer deployment.
\Cref{sec:architecture} develops these relations and the buffer transition
underlying this design.

\section{The HLCP Architecture}
\label{sec:architecture}

HLCP keeps part of a paired reserve position outside the active trading
curve and temporarily deploys it when a submitted swap projects a large
price movement.  This design requires both an understanding of execution
at reduced depth and a reserve transition that supplies additional depth
from the buffer.
\Cref{sec:finite-order-geometry} derives the effect of liquidity on the
trader's average execution price and the pool's marginal price after the swap.
Building on this geometry, \cref{sec:physical-scaling} defines the active pool
and paired buffer, and \cref{sec:prospective-activation} specifies the reserve
deposit triggered by a submitted swap.  \Cref{sec:stateful-recovery} completes
the transition by returning the activated share to the buffer in the token
composition left by the swap.

\subsection{Finite-order geometry}
\label{sec:finite-order-geometry}

Reducing the active reserve base affects both the trader's average execution
price and the pool's marginal price after the swap.  We derive how these two
outcomes depend on liquidity to establish the geometric basis for HLCP's
reserve split.  To isolate this dependence, we keep the swap size and starting
price fixed while varying active liquidity.  \Cref{fig:geometry} illustrates
the corresponding secant and endpoint tangents on the reserve curve.
We first consider a fee-free CPMM at pre-swap price $P_0=y/x$ and liquidity
$L=\sqrt{xy}$.  Selling $q>0$ units of $x$ produces
\begin{equation}
  \Delta y=y-\frac{L^2}{x+q},
  \qquad
  P_{\mathrm{eff}}=\frac{\Delta y}{q},
  \qquad
  P_1=\frac{L^2}{(x+q)^2}.
  \label{eq:geometry-sale}
\end{equation}
Here, $P_0$ and $P_1$ are the marginal prices before and after this swap,
corresponding to $P_{\mathrm{amm}}$ and $P_{\mathrm{post}}$ in
\cref{sec:cpmm-background}.  In \cref{fig:geometry}, the straight line joining
the pre-swap and post-swap reserve states is the secant.  Its slope is
negative because the pool receives $x$ and pays out $y$; the slope magnitude
$\Delta y/q$ gives the trader's average execution price $P_{\mathrm{eff}}$.
The tangent at either state also has a negative slope, whose magnitude
$y/x$ gives the marginal price at that state.  The external market valuation
$P_{\mathrm{ext}}$ is a separate price reference rather than an outcome of
this swap.

\begin{figure}[htbp]
  \centering
  \resizebox{0.88\linewidth}{!}{\begin{tikzpicture}[
  x=5.15cm,
  y=3.55cm,
  font=\footnotesize,
  line cap=round,
  line join=round
]
  \draw[-{Stealth[length=2mm]},semithick]
    (0.58,0.36) -- (2.25,0.36)
    node[below right=-1pt] {risky reserve $x$};
  \draw[-{Stealth[length=2mm]},semithick]
    (0.61,0.33) -- (0.61,1.76)
    node[above left=-1pt,rotate=90] {numeraire reserve $y$};
  \draw[black,very thick]
    plot[domain=0.65:1.84,samples=180,smooth] (\x,{1/\x});
  \node[anchor=south west,inner sep=0pt]
    at (0.86,1.23) {$xy=L^2$};

  \coordinate (s0) at (1.00,1.00);
  \coordinate (s1) at (1.48,0.6756757);

  \draw[hlcpblue,very thick] (0.86,1.14) -- (s0);
  \node[anchor=east,hlcpblue,inner sep=0pt]
    at (0.84,1.07) {$|m|=P_0$};

  \draw[hlcporange,very thick] (s1) -- (1.66,0.5935);
  \node[anchor=north,hlcporange,inner sep=0pt]
    at (1.62,0.548) {$|m|=P_1$};

  \draw[hlcpgray,very thick] (s0) -- (s1);
  \node[anchor=south,align=center,inner sep=0pt]
    at (1.24,0.91) {secant\\[-1pt]$|m|=P_{\mathrm{eff}}$};

  \draw[-{Stealth[length=2mm]},black,semithick]
    plot[domain=1.14:1.38,samples=45,smooth] (\x,{1/\x});
  \draw[semithick] (1.00,0.385) -- (1.00,0.445);
  \draw[semithick] (1.48,0.385) -- (1.48,0.445);
  \draw[{Stealth[length=1.5mm]}-{Stealth[length=1.5mm]},semithick]
    (1.00,0.415) -- (1.48,0.415);
  \node[anchor=south,inner sep=0pt] at (1.24,0.455) {$q=\Delta x$};

  \node[circle,draw=hlcpblue,very thick,fill=white,inner sep=2.1pt] at (s0) {};
  \node[rectangle,draw=hlcporange,very thick,fill=hlcporange,inner sep=2.0pt] at (s1) {};
  \node[anchor=north east,inner sep=0pt]
    at ($(s0)+(-0.025,-0.018)$) {$s_0$: pre};
  \node[anchor=south west,inner sep=0pt]
    at ($(s1)+(0.020,0.035)$) {$s_1$: post};

  \draw[hlcpgreen,densely dashed,very thick]
    (0.61,1.73) --
    node[pos=0.42,above=2.5pt,sloped,hlcpgreen,inner sep=0pt]
      {external valuation: $|m|=P_{\mathrm{ext}}$}
    (2.185,0.36);
\end{tikzpicture}}
  \caption{Fee-free geometry of a sell-$x$ order from reserve state $s_0$ to
  $s_1$.  The slope magnitude $|m|$ gives the average execution price
  $P_{\mathrm{eff}}$ for the secant and the marginal prices $P_0$ and $P_1$
  for the pre- and post-swap tangents.  The dashed line represents the
  independent external valuation with slope $-P_{\mathrm{ext}}$; its
  intercept is illustrative.}
  \Description{A constant-product reserve curve with distinct pre- and
  post-order state markers joined by a chord and a curved order arrow. Short
  colored strokes show the endpoint marginal slopes, while a separate dashed
  line spanning the visible reserve axes shows the independent external
  valuation slope.}
  \label{fig:geometry}
\end{figure}

Set $K=q\sqrt{P_0}$, which expresses the order
size in the same units as liquidity $L$.  Holding the input and starting
price fixed keeps $K$ fixed, so we can study how execution changes with $L$.
For this sale of $x$, the average execution price $P_{\mathrm{eff}}$ lies
below the starting marginal price $P_0$, the rate available to an
infinitesimally small swap.  Output slippage measures this relative shortfall:
\begin{equation}
  S(L;K)
    =1-\frac{P_{\mathrm{eff}}}{P_0}
    =\frac{K}{L+K}.
  \label{eq:slippage-kernel}
\end{equation}
For fixed $K>0$,
\begin{equation}
  \frac{\partial S}{\partial L}
    =-\frac{K}{(L+K)^2}<0,
  \qquad
  \frac{\partial^2 S}{\partial L^2}
    =\frac{2K}{(L+K)^3}>0.
  \label{eq:slippage-derivatives}
\end{equation}
The negative first derivative shows that increasing active liquidity reduces
output slippage.  The positive second derivative shows that the slope becomes
less negative as liquidity increases, so each additional unit removes less
slippage than the previous one.
For the fixed swap, slippage approaches zero as active liquidity grows
without bound, and the average execution price approaches the starting
marginal price.

The secant and the post-swap tangent in \cref{fig:geometry} describe different
outcomes of the same order: the average price received by the trader and the
marginal price left in the pool.  Expressing both relative to the pre-swap
price gives
\begin{equation}
  \frac{P_{\mathrm{eff}}}{P_0}=1-S,
  \qquad
  \frac{P_1}{P_0}=(1-S)^2,
  \qquad
  1-\frac{P_1}{P_0}=S(2-S)>S.
  \label{eq:secant-tangent-sale}
\end{equation}
For a sale of the first token, the pool's relative price decline is therefore
larger than the trader's output slippage.  Average execution and the price
left in the pool are distinct outcomes of the same swap, even though both
improve with greater depth.

For a swap that buys $x$, $S$ instead measures the relative shortfall of the
average execution rate from the initial marginal rate, both expressed in
units of $x$ per unit of $y$.  The pool prices $P_0$ and $P_1$ retain the
convention of $y$ per unit of $x$, so their ratio now increases:
\begin{equation}
  \frac{P_1}{P_0}=(1-S)^{-2},
  \qquad
  \frac{P_1}{P_0}-1
    =\frac{S(2-S)}{(1-S)^2}.
  \label{eq:secant-tangent-buy}
\end{equation}
The change in the pool's marginal price therefore depends on the swap
direction, even at the same execution slippage.  In both directions,
increasing active liquidity brings the average execution rate and the
post-swap marginal rate toward their corresponding pre-swap rates.  The
normalized output shortfall measures execution deterioration in either
direction, while the pool-price ratio distinguishes a rise from a decline.

The secant relation quantifies the execution cost of keeping less liquidity
active, which Experiment~1 uses to calibrate the initial reserve split.
HLCP increases this liquidity through a buffer deposit when the preview of
a submitted swap reveals stress.

\subsection{N-scaled active pool and paired buffer}
\label{sec:physical-scaling}

HLCP implements the reserve split by scaling both tokens together,
preserving their initial reserve ratio.  Fix $N\in(0,1]$, an initial
full-depth liquidity $L_{\mathrm{tot},0}$, and price $P_0$.  HLCP places the chosen fraction of
the reserves in the active CPMM:
\begin{equation}
  L_{\mathrm{a},0}=NL_{\mathrm{tot},0},
  \qquad
  x_{\mathrm{a},0}=\frac{L_{\mathrm{a},0}}{\sqrt{P_0}},
  \qquad
  y_{\mathrm{a},0}=L_{\mathrm{a},0}\sqrt{P_0}.
  \label{eq:physical-n-scaling}
\end{equation}
Full depth is the benchmark in which all paired reserves are active.  HLCP
divides this same capital between the active CPMM and the buffer, so its
liquidity remains backed by physical tokens.  The remaining $1-N$ share of
each token initializes the paired buffer.  For example, if the full-depth
pool holds 100 units of $x$ and 200 units of $y$, choosing $N=0.8$ leaves
80 units of $x$ and 160 units of $y$ active, with the remaining 20 and 40
units in the buffer.  The starting price remains 2 units of $y$ per unit of
$x$, while active liquidity is 80\% of full depth.  The policy coefficient $N$
sets the initial active fraction and also enters the controller.  As swaps
and recovery change the token balances, those balances define the next
event's state, and the actual active share can depart from its initial value.

Immediately before event $t$, let active reserves be $(x_t,y_t)$, buffer
balances be $(B^x_t,B^y_t)$, and $P_t=y_t/x_t$.  The subscript $t$ identifies
the current event, while the subscript 0 marks initialization.  Homothetic expansion requires
both tokens in the current reserve ratio.  The amount available for this
expansion is therefore limited by the two buffer balances:
\begin{equation}
  C_t
    =\min\!\left\{B^x_t\sqrt{P_t},
                  \frac{B^y_t}{\sqrt{P_t}}\right\}.
  \label{eq:paired-capacity}
\end{equation}
Each term converts one of the buffer's token balances into the amount of
liquidity it can support at the current reserve ratio.  The smaller amount
sets the paired capacity $C_t$, expressed in the liquidity units of $L$,
because a larger deposit would
require more of the limiting token than the buffer holds.  Tokens
that cannot be paired in this ratio remain part of LP wealth.  Deploying
them would require a separate token conversion, which lies outside the HLCP
transition.

\subsection{State-conditioned homothetic expansion}
\label{sec:prospective-activation}

The geometry above distinguishes the trader's average execution price from
the pool price left by a swap.  HLCP bases buffer activation on the projected
change in this pool price.  Before executing submitted order $o_t$, HLCP
calculates the marginal price that the order would produce without buffer
activation.  Step~1 in \cref{fig:prospective-transition}(a) illustrates this
fee-aware preview: it produces a hypothetical post-swap state while leaving
the pre-swap reserve state $s_t$ unchanged.  Let
$T_{f_t}$ denote the fee-retaining swap transition in
\cref{sec:cpmm-background}, with input fee rate $f_t$.  A hat marks a
hypothetical value, while the subscript $t+$ identifies the post-swap stage.
Stress $\phi_t$ is zero when
the magnitude of the projected relative price move stays within the protocol
tolerance $\tau$, and equals the excess above that tolerance otherwise:
\begin{equation}
  (\widehat x_{t+},\widehat y_{t+})
    =T_{f_t}(x_t,y_t;o_t),
  \qquad
  \widehat P_{t+}=\frac{\widehat y_{t+}}{\widehat x_{t+}},
  \qquad
  \phi_t
    =\max\!\left\{0,
      \left|1-\frac{\widehat P_{t+}}{P_t}\right|-\tau\right\}.
  \label{eq:prospective-stress}
\end{equation}
The preview uses the current reserves, swap fee, and submitted order.  The
protocol tolerance $\tau$ controls when the projected price move triggers
activation; it is chosen as a controller parameter rather than derived from
an arbitrage-profit condition.

The controller uses projected stress to select a deposit within the available
paired capacity.  We choose a smooth, saturating response so that buffer
deployment increases with projected stress and approaches, but does not
exceed, the available
paired capacity.  The product $N\alpha$ determines the response strength;
for a fixed initial active fraction $N$, $\alpha$ tunes how rapidly
deployment approaches this limit as stress increases.
\begin{equation}
  \Delta C_t
    =\frac{N\alpha\phi_tC_t}{1+N\alpha\phi_t},
  \qquad \alpha\geq0.
  \label{eq:single-controller}
\end{equation}
Here, $\Delta C_t$ denotes the liquidity deployed for this event.
When $\Delta C_t>0$, the mechanism transfers
\begin{equation}
  \Delta x_{C,t}=\frac{\Delta C_t}{\sqrt{P_t}},
  \qquad
  \Delta y_{C,t}=\Delta C_t\sqrt{P_t}
  \label{eq:homothetic-vector}
\end{equation}
from the buffer to the active pool.  Steps~2 and~3 in
\cref{fig:prospective-transition}(b) distinguish this deposit from the swap
that follows.  In step~2, adding both tokens at the current reserve ratio
moves the state outward along the same price ray.  The pool gains depth while
keeping its pre-swap marginal price.  In step~3, the submitted order settles
exactly once against the expanded reserves.
Activation is determined by the preview and remains fixed throughout
settlement; the realized post-swap price does not revise the decision.

\begin{figure*}[t]
  \centering
  \begin{tikzpicture}[
  x=5.2cm,
  y=3.85cm,
  font=\footnotesize,
  line cap=round,
  line join=round,
  action/.style={very thick,-{Stealth[length=1.9mm]}},
  step/.style={circle,draw=#1,very thick,fill=white,inner sep=1.4pt,
               text=#1,font=\bfseries\footnotesize}
]
  \begin{scope}
    \draw[-{Stealth[length=1.7mm]},semithick]
      (0.27,0.16) -- (1.38,0.16) node[below right=-1pt] {$x$};
    \draw[-{Stealth[length=1.7mm]},semithick]
      (0.30,0.13) -- (0.30,1.18) node[above left=-1pt] {$y$};
    \draw[hlcpgray,very thick]
      plot[domain=0.34:1.31,samples=150,smooth] (\x,{0.24/\x});

    \coordinate (qa) at ({sqrt(0.24)},{sqrt(0.24)});
    \coordinate (qhat) at ({sqrt(0.24)+0.38},{0.24/(sqrt(0.24)+0.38)});
    \draw[hlcpblue,densely dashed,action]
      plot[domain=0.505:0.855,samples=60,smooth] (\x,{0.24/\x});
    \node[circle,draw=black,fill=black,inner sep=1.9pt] at (qa) {};
    \node[anchor=east,inner sep=0pt]
      at ($(qa)+(-0.035,-0.015)$) {$s_t$};
    \node[hlcpgray,fill=white,inner sep=0.5pt] at (qhat) {$\times$};
    \node[anchor=north west,inner sep=0pt]
      at ($(qhat)+(0.025,-0.018)$) {$\widehat s_{t+}$};

    \node[step=hlcpblue] at (0.36,1.00) {1};
    \node[anchor=west,inner sep=0pt]
      at (0.43,1.00)
      {Preview $o_t$: $\widehat P_{t+}\!\rightarrow\!\phi_t\!\rightarrow\!\Delta C_t$};
    \node[anchor=west,text=hlcpgray,inner sep=0pt]
      at (0.58,0.69) {physical state remains $s_t$};

    \node[font=\small\bfseries] at (0.82,1.26)
      {(a) Counterfactual AMM quote};
  \end{scope}

  \begin{scope}[xshift=7.40cm]
    \draw[-{Stealth[length=1.7mm]},semithick]
      (0.27,0.16) -- (1.38,0.16) node[below right=-1pt] {$x$};
    \draw[-{Stealth[length=1.7mm]},semithick]
      (0.30,0.13) -- (0.30,1.18) node[above left=-1pt] {$y$};
    \draw[hlcpgray,thick]
      plot[domain=0.34:1.31,samples=150,smooth] (\x,{0.24/\x});
    \draw[hlcporange,very thick]
      plot[domain=0.64:1.31,samples=150,smooth] (\x,{0.6084/\x});
    \node[anchor=west,align=left,hlcporange,inner sep=0pt]
      at (0.94,0.82) {homothetically\\expanded curve};

    \coordinate (pa) at ({sqrt(0.24)},{sqrt(0.24)});
    \coordinate (pplus) at (0.78,0.78);
    \coordinate (psettle) at (1.16,{0.6084/1.16});
    \pgfmathsetmacro{\recoveryscale}{sqrt(0.24)/0.78}
    \coordinate (precover) at
      ({1.16*\recoveryscale},{(0.6084/1.16)*\recoveryscale});

    \draw[hlcpblue,densely dotted,thick] (0.34,0.34) -- (0.88,0.88);
    \draw[hlcpgray,densely dotted,thick] (0.34,0.1537) -- (1.30,0.5878);

    \draw[hlcpblue,action]
      ($(pa)+(0.014,0.015)$) -- ($(pplus)+(-0.017,-0.018)$);
    \draw[hlcporange,action]
      plot[domain=0.795:1.145,samples=65,smooth] (\x,{0.6084/\x});
    \draw[hlcpgray,action]
      ($(psettle)+(-0.012,-0.005)$) -- ($(precover)+(0.015,0.007)$);

    \node[circle,draw=black,fill=black,inner sep=1.9pt] at (pa) {};
    \node[circle,draw=hlcpblue,very thick,fill=white,inner sep=2.1pt] at (pplus) {};
    \node[rectangle,draw=hlcporange,very thick,fill=hlcporange,inner sep=2.1pt]
      at (psettle) {};
    \node[diamond,draw=hlcpgray,fill=white,inner sep=1.9pt] at (precover) {};

    \node[step=hlcpblue] at (0.37,1.02) {2};
    \node[anchor=west,inner sep=0pt]
      at (0.44,1.02) {Deposit $\Delta C_t$};

    \node[step=hlcporange] at (0.87,1.02) {3};
    \node[anchor=west,inner sep=0pt]
      at (0.94,1.02) {Settle $o_t$ once};

    \node[step=hlcpgray] at (0.96,0.35) {4};
    \node[anchor=west,inner sep=0pt]
      at (1.03,0.35) {Recover $\omega_t$};

    \node[font=\small\bfseries] at (0.82,1.26)
      {(b) Physical transition};
  \end{scope}
\end{tikzpicture}
  \caption{One HLCP event, illustrated with fee-free reserve geometry.
  (a) Previewing the submitted order $o_t$ gives the hypothetical post-swap
  state $\widehat s_{t+}$, while the actual reserves remain at $s_t$.
  (b) The backed buffer deposit $\Delta C_t$ moves the reserve state outward
  along the $P_t$ ray.  The same $o_t$ then settles once on the homothetically
  expanded curve, after which the activated share $\omega_t$ returns to the
  buffer along the post-swap price ray.}
  \Description{Two panels use fee-free reserve geometry to separate a
  counterfactual AMM quote from the physical transition. Numbered arrows show a counterfactual AMM quote,
  price-preserving homothetic activation, one settlement of the same order,
  and pro-rata recovery.}
  \label{fig:prospective-transition}
\end{figure*}

\begin{proposition}[Bounded price-preserving activation]
\label{prop:hlcp-core}
For finite $\phi_t\geq0$, $C_t>0$, $N\in(0,1]$, and $\alpha>0$, the
controller is zero at zero stress, strictly increasing and strictly concave
in $\phi_t$, and satisfies $0\leq\Delta C_t<C_t$.  The deposit in
\cref{eq:homothetic-vector} preserves $P_t$ and raises physical CPMM liquidity
by exactly $\Delta C_t$.
\end{proposition}

\begin{proof}
Differentiating \cref{eq:single-controller} gives
\[
  \frac{\partial\Delta C_t}{\partial\phi_t}
    =\frac{N\alpha C_t}{(1+N\alpha\phi_t)^2}>0,
  \qquad
  \frac{\partial^2\Delta C_t}{\partial\phi_t^2}
    =-\frac{2(N\alpha)^2C_t}{(1+N\alpha\phi_t)^3}<0.
\]
Moreover,
$C_t-\Delta C_t=C_t/(1+N\alpha\phi_t)>0$.  Writing
$(x_t,y_t)=(L_t/\sqrt{P_t},L_t\sqrt{P_t})$, the post-deposit state is
\[
  \left(\frac{L_t+\Delta C_t}{\sqrt{P_t}},
        (L_t+\Delta C_t)\sqrt{P_t}\right),
\]
which retains reserve ratio $P_t$ and has liquidity $L_t+\Delta C_t$.
\end{proof}

Each activation thus adds fully backed depth while preserving the current
price.  The subsequent swap changes the token composition of the temporary
share, so recovery must return the balances left by the swap and carry that
composition into the next event.

\subsection{Settlement and recovery}
\label{sec:stateful-recovery}

Recovery returns the temporary contribution to the buffer after it has
provided depth for settlement.  Because both tokens were deposited in the
pool's reserve ratio, they represent the same proportional share of the
expanded reserves.  Let $L_t=\sqrt{x_ty_t}$ before deposit and set
$c_t=\Delta C_t>0$ as shorthand for the deposited liquidity in the recovery
calculation.
The total liquidity immediately after deposit is $L_t+c_t$, so the temporary
contribution represents the following conceptual LP share:
\begin{equation}
  \omega_t=\frac{c_t}{L_t+c_t}.
  \label{eq:activated-share}
\end{equation}
After the submitted order settles, HLCP withdraws this same fraction of each
token from the active reserves, including the fees retained during the swap,
and returns those tokens to the buffer.  Step~4 in
\cref{fig:prospective-transition}(b) shows this withdrawal within the full
event, while \cref{fig:proportional-recovery} shows its geometry about the
origin.  In that figure, $(x_t^{\mathrm{tr}},y_t^{\mathrm{tr}})$ denotes the
active reserves immediately after settlement, and $(x_{t+1},y_{t+1})$ denotes
those remaining after recovery, which form the next event's pre-state.
Removing the share $\omega_t$ leaves the same fraction
$1-\omega_t$ of each post-swap reserve in the active pool.  The reserve state
therefore moves inward along the post-swap price ray, preserving the reserve
ratio and the price established by settlement.  This is the inward version
of the homothety in \cref{fig:homothety-overview}, now applied to the token
balances left by the order.

\begin{figure}[htbp]
  \centering
  \resizebox{0.74\linewidth}{!}{
\begin{tikzpicture}[
  x=2.05cm,
  y=1.45cm,
  font=\footnotesize,
  line cap=round,
  line join=round
]
  \draw[-{Stealth[length=1.7mm]},semithick]
    (0,0) -- (4.65,0) node[below right=-1pt] {$x$};
  \draw[-{Stealth[length=1.7mm]},semithick]
    (0,0) -- (0,2.45) node[above left=-1pt] {$y$};
  \node[anchor=north east,inner sep=1pt] at (0,0) {$O$};

  \draw[hlcpgray,thick,densely dashed]
    plot[domain=0.99:4.40,samples=180,smooth] (\x,{2.25/\x});
  \draw[hlcpblue,very thick]
    plot[domain=0.56:4.40,samples=180,smooth] (\x,{1.265625/\x});

  \draw[hlcpgray,densely dotted,semithick] (0,0) -- (2.50,0.90);
  \draw[hlcpblue,very thick,-{Stealth[length=1.8mm]}]
    (2.46,0.8856) -- (1.915,0.6894);
  \node[rectangle,draw=hlcpgray,fill=hlcpgray,inner sep=1.9pt]
    at (2.50,0.90) {};
  \node[circle,draw=hlcpblue,fill=white,very thick,inner sep=2.0pt]
    at (1.875,0.675) {};
  \node[anchor=west,inner sep=0pt] at (2.60,1.10)
    {$(x_t^{\mathrm{tr}},y_t^{\mathrm{tr}})$};
  \node[anchor=west,text=hlcpblue,inner sep=0pt] at (1.95,0.27)
    {$(x_{t+1},y_{t+1})$};

  \draw[hlcpgray,densely dotted,thin,opacity=0.55]
    (0,0) -- (1.39,2.0016);
  \draw[hlcpgray,semithick,opacity=0.65,-{Stealth[length=1.5mm]}]
    (1.229,1.76976) -- (0.9585,1.38024);
  \node[circle,draw=hlcpgray,fill=white,thin,inner sep=1.3pt]
    at (1.25,1.80) {};
  \node[circle,draw=hlcpgray,fill=white,thin,inner sep=1.3pt]
    at (0.9375,1.35) {};

  \node[align=center,inner sep=0pt] at (3.16,1.93)
    {Both reserves scale by $1-\omega_t$\\[2pt]
     post-swap price unchanged};

  \draw[hlcpgray,thick,densely dashed] (0.35,2.66) -- (0.65,2.66);
  \node[anchor=west,inner sep=0pt] at (0.75,2.66)
    {Before recovery};
  \draw[hlcpblue,very thick] (2.65,2.66) -- (2.95,2.66);
  \node[anchor=west,inner sep=0pt] at (3.05,2.66)
    {After recovery};
\end{tikzpicture}}
  \caption{Proportional recovery after settlement.
  Withdrawing the share $\omega_t$ scales both post-swap reserves by
  $1-\omega_t$, moving the state inward on the same post-swap price ray.}
  \Description{A post-swap constant-product reserve curve and a smaller
  curve after proportional recovery are shown with the origin visible.
  An inward arrow joins the actual post-swap state to the remaining active
  state on the same ray, preserving the price established by settlement.}
  \label{fig:proportional-recovery}
\end{figure}

The returned token mix generally differs from the mix originally deposited
because the pool has received one
token and paid out the other during settlement.  The next event uses the
updated active and buffer balances and recomputes paired capacity from the
realized buffer composition using \cref{eq:paired-capacity}.
\Cref{app:mechanism-accounting} gives the exact token updates for the full
sequence.

\section{Evaluation Methodology}
\label{sec:method}

\Cref{sec:architecture} derives how active liquidity affects swap execution
and defines how HLCP temporarily restores depth through a paired-buffer
deposit.  Evaluating this design requires connecting the initial reserve
split to the orders that the pool must execute: leaving capital in the
buffer reduces the depth initially available, while deploying it for a large
order can restore that depth.  We examine these two effects in sequence
using historical market data.  \Cref{sec:method-exp1} uses a year of observed
Swaps to calibrate the initial active fraction, describing the data and
sampling, the normalization of order size, and the comparison with full
depth.  \Cref{sec:method-exp2} then carries $N^*$ unchanged into an acute-order
case, explaining how an external price event and a historical V3 state
define the common order, the three comparators, and the execution measures.
Together, the experiments evaluate the execution cost of the initial
reserve split and the benefit of activating its existing buffer when the
submitted order requires greater depth.

\subsection{Annual capital-scale calibration}
\label{sec:method-exp1}

Experiment~1 evaluates the execution cost of keeping part of the initial
capital in the buffer before any activation.  For observed order sizes, we
compare an $N$-scaled balanced synthetic CPMM with its full-depth counterpart
using the same gross input and fee.  Since the reserve split is expressed
as a capital fraction, the size of an order must also be measured relative
to the capital held by the pool at that time.  ETH prices and pool sizes
change over the year, so the same token amount can represent different
fractions of pool capital.  We therefore normalize each input value by the
contemporaneous value of the pool's token holdings.  The following
development first describes how the required observations are obtained,
then defines this common scale and uses it to calculate the execution gap
as the active fraction changes.

\paragraph{Data and sampling.}
The comparison requires each order's actual gross input together with the
pool's price and token balances immediately before execution.  We obtain
these quantities from 2025 Swap records and historical Uniswap V3 states,
using WETH/USDC so that both the input and pool holdings can be valued in
USDC.  Among Factory-created WETH/USDC pools, we select the pool with the
largest sum of absolute raw-USDC Swap amounts over the year.  The selected
5-bp pool contained 2,216,127 Swap events, all of which form the annual
population.  We obtain the annual logs through Etherscan V2 and retrieve
historical states for the sampled candidates through Alchemy archive RPC
\cite{etherscanv2api,alchemydocs}.

The complete annual logs are used for pool selection and define the
population from which candidates are drawn.  Evaluating an individual
candidate additionally requires reconstructing its input and exact
pre-swap state from historical states and transaction records.  We use a
fixed pseudorandom sample of 500 events for this detailed reconstruction,
with the sample size set in advance to bound the reconstruction workload.
Candidates are selected without replacement before checking reconstruction
eligibility or calculating their normalized order sizes and execution
gaps.  This sequence fixes the candidate panel before the quantities used
in calibration are evaluated.

For each candidate, we reconcile the signed Swap amounts with token
transfers into and out of the pool to identify the actual gross input.
The input and pool holdings must refer to the same point in execution so
that the normalized order size describes the capital available when that
Swap begins.  We retrieve the preceding block's pool state and use
pool-event and token-transfer records to establish when it also represents
the state immediately before the target Swap.  Of the 500 sampled events,
283 met these reconstruction conditions, including 125 USDC-in and
158 WETH-in Swaps.  Earlier pool activity in the same block accounted for
206 exclusions because it failed the pre-state matching condition.
Excluded events were not replaced.  Each eligible Swap supplies a separate
snapshot for the synthetic comparison.

\paragraph{Price and capital normalization.}
We place the reconstructed orders on the common capital-relative scale by
valuing both the gross input and pool holdings at the same pre-swap price.
For eligible Swap $i$, let $P_i$ be the exact pre-swap USDC-per-WETH price
obtained from the V3 pool's price state.  Let $x_i$ and $y_i$ be the WETH and
USDC held by the pool after subtracting accrued protocol fees.  Their
marked value is
\begin{equation}
  W_i=P_i x_i+y_i.
  \label{eq:capital-proxy}
\end{equation}
We call $W_i$ the \emph{non-protocol pool-custodied marked-capital proxy}.
The gross input notional and its normalized size are
\begin{equation}
  A_i^{\mathrm{notional}}=
  \begin{cases}
    \text{gross USDC input}, & \text{USDC in},\\
    P_i\,\text{gross WETH input}, & \text{WETH in},
  \end{cases}
  \qquad
  z_i=\frac{A_i^{\mathrm{notional}}}{W_i}.
  \label{eq:event-normalization}
\end{equation}
Thus $z_i$ expresses the input value as a fraction of the pool's marked
capital immediately before that Swap.

To isolate the execution effect of reducing active capital, we use the
observed price $P_i$ and marked-capital scale $W_i$ to construct a balanced
synthetic CPMM.  Its two reserves have equal value at $P_i$, giving the
full-depth reference
\begin{equation}
  x_i^F=\frac{W_i}{2P_i},
  \qquad
  y_i^F=\frac{W_i}{2},
  \label{eq:balanced-reconstitution}
\end{equation}
and the corresponding $N$-scaled reserves are $(Nx_i^F,Ny_i^F)$.
The balanced token amounts need not match the observed V3 balances.
Each comparison thus uses the observed order size, starting price, and
marked-capital scale while holding the synthetic reserve geometry fixed.

\paragraph{Execution gap and calibration.}
With these full-depth and scaled pools defined, we measure the output lost
when only a fraction $N$ remains active.  Both pools execute the same gross
input with the same fee, so their output difference isolates the effect of
the smaller active reserves.  With input-fee multiplier $\gamma=1-f$, the exact incremental
output deterioration relative to the full synthetic pool is, in either
direction,
\begin{equation}
  g_i(N)
    =10{,}000\left(1-\frac{\mathcal{O}_i(N)}{\mathcal{O}_i^F}\right)
    =10{,}000\,
      \frac{2\gamma z_i(1-N)}{N+2\gamma z_i}\quad\text{bp},
  \label{eq:event-gap}
\end{equation}
where $\mathcal{O}_i^F$ and $\mathcal{O}_i(N)$ are the outputs from the full
and scaled curves for the same gross input and fee.  Thus $g_i(N)$ measures
the additional output loss caused by reducing active capital.

For each active fraction $N$, we summarize the additional output loss across
eligible Swaps using the 95th percentile (P95).  At least 95\% of these sampled
Swaps have losses no greater than this value.  Each Swap receives equal
weight.  If $r$ is their count and the normalized
sizes are ordered from smallest to largest, the no-interpolation Type-1
percentile is
$z_{.95}=z_{(\lceil .95r\rceil)}$, and monotonicity in $z_i$ gives the exact
P95 gap curve in basis points,
\begin{equation}
  G_{.95}(N)
    =10{,}000\,
      \frac{2\gamma z_{.95}(1-N)}{N+2\gamma z_{.95}},
  \qquad 0<N\leq1.
  \label{eq:p95-gap}
\end{equation}
For a descriptive calibration, we use the selected pool's 5-bp input fee as
a reference scale for this additional output loss.  We define
$N^*$ as the point where the P95 curve reaches that fee scale,
$T_{\mathrm{fee}}=10{,}000f$ bp.  Equivalently, with
$t_{\mathrm{fee}}=f$,
\begin{equation}
  N^*
    =\frac{2\gamma z_{.95}(1-t_{\mathrm{fee}})}
           {2\gamma z_{.95}+t_{\mathrm{fee}}}.
  \label{eq:n-star}
\end{equation}
At this fraction, at least 95\% of the eligible sampled Swaps have an
additional output loss no greater than the 5-bp threshold.

\subsection{Acute-tail three-arm comparison}
\label{sec:method-exp2}

Experiment~1 determines an initial reserve split from observed order sizes;
Experiment~2 examines how the capital left in its buffer can restore depth
for a large order.  We retain $N^*$ without recalibration and apply a common
order to HLCP, a matched static CPMM, and a historical passive V3 pool.
The order is constructed from an external price event and the historical
V3 state before evaluating HLCP's response.  We first explain this
event-to-order construction, then define the comparators and the measures
used to examine the trader's execution and the pool's price path.  This
ordering connects the buffer allocation from the annual calibration to
the execution effect of deploying it in the acute case.

\paragraph{External event and Uniswap state.}
To set the order's scale from historical market conditions, we take the
acute ETH--USDC move from Coinbase Advanced Trade prices and translate it
into an order using historical Uniswap V3 state.  Selecting the move in an
external market separates event selection
from the V3 comparator's tick response.  We select the largest relative
open-to-close price change among the provider-published 15-minute candles in
2025.  Specifically, the selection rule maximizes $\max\{C/O,O/C\}$,
where $O$ and $C$ are the exact open and close prices
\cite{coinbaseadvancedapi}.  The 35,015 available candles are used as
published; the 25 missing intervals are left unfilled.  The comparison uses
the selected acute event, with the historical pool snapshot, common order,
and controller parameters held fixed.

We apply the selected candle's relative price move to the pool price
immediately before the event to define the target price.  The pool state and
initialized ticks are read at the
same canonical block, the last one before the candle begins.  Exact V3
integer tick arithmetic then determines the smallest gross exact-input
order $Q$ that reaches the mapped endpoint.  We check the output, terminal
price, and tick crossings against a historical QuoterV2 quote.  The external
candle therefore sets the shock magnitude, while the observed V3 liquidity
profile determines the common order applied to all three pools.

\paragraph{Comparators and path metrics.}
We use the matched static CPMM to isolate the effect of activation and
historical V3 to compare with an observed passive-liquidity allocation.
The same starting price $P_0$, capital
proxy $W_0$, direction, fee, and gross
order $Q$ are applied once to three arms:
\begin{enumerate}
  \item the historical passive V3 multi-tick state;
  \item a static balanced synthetic CPMM initialized at Experiment~1's $N^*$;
  \item HLCP with the same initial $N^*$ and the fixed controller
        $\alpha=100$, $\tau=0.05$.
\end{enumerate}
The static CPMM and HLCP begin with the same balanced reserves and paired
buffer.  The static arm keeps its buffer dormant, while HLCP previews the
whole order and transfers the resulting paired deposit to the active
reserves before settlement.  The order then executes once on the expanded
curve.  Their comparison therefore isolates the effect of activation under
the same initial capital, fee, and order.  Passive V3 retains its observed
token balances and distribution of liquidity across price ranges.  It shares
the starting price and marked-capital scale of the synthetic pools, whose
initial total reserves are balanced as in \cref{eq:balanced-reconstitution}.

The restored depth affects both the output received by the trader and the
price left in the pool.  We therefore measure the trader's average execution
price together with the pool's price movement as the common order fills.
We follow the cumulative execution of this
one order, beginning after the full-order preview and buffer deposit for
HLCP.  For the observed downward
direction, let $q_{\mathrm{WETH}}$ be cumulative gross WETH input, written
as $q$ below, and let $P(u)$ be the resulting marginal USDC/WETH price.
Holding the pre-event price $P_0$ and capital $W_0$ fixed, we define
normalized input $u$ and natural-log price displacement $p(u)$ as
\begin{equation}
  u=\frac{P_0\,q_{\mathrm{WETH}}}{W_0},
  \qquad
  p(u)=\log\frac{P_0}{P(u)}.
  \label{eq:exp2-path-normalization}
\end{equation}
\Cref{fig:exp2-paths}(a) displays the same marginal-price path as a
percentage decline from $P_0$.
The volume-weighted average execution price (VWAP) is the USDC output divided
by the gross WETH input.  We measure its deterioration relative to the
starting price as
\begin{equation}
  1-\frac{\text{USDC output}}
           {\text{gross WETH input}\cdot P_0}.
  \label{eq:down-vwap}
\end{equation}
To follow average execution through the order, we evaluate this measure
using cumulative USDC output and cumulative gross WETH input at each
positive input amount, reported as a percentage in \cref{fig:exp2-paths}(b).
To identify parts of the path with less capacity to absorb input, we also
measure local depth and the price response to a small input increment.  We
report whether the order fills completely, its terminal price
displacement, the minimum balanced-equivalent local depth, and the worst
local price sensitivity along the executed path.  The last quantity is the
reciprocal of local input capacity.  V3's active in-range liquidity is
converted from raw units to token units as $L_h=L_{\mathrm{raw}}/10^{12}$.
For the two synthetic arms, $L_h$ is the liquidity scale of the active
reserves at that point on the path, as defined in \cref{eq:cpmm-state}, and is
already in token units.  At price $P$,
normalized balanced-equivalent local depth is
\begin{equation}
  D(P)=\frac{2L_h(P)\sqrt P}{W_0}.
  \label{eq:exp2-local-depth}
\end{equation}
For this downward path, local input capacity measures gross normalized input
per unit absolute log-price change.  It is
\begin{equation}
  \kappa_{V3}(P)=
    \frac{P_0L_h(P)}{2(1-f)W_0\sqrt P}
  \label{eq:exp2-v3-capacity}
\end{equation}
for V3.  For the static CPMM or HLCP path, with fees retained in the reserves,
cumulative gross WETH input $q$ and settlement-start WETH reserve $x_s$
(after any buffer deposit) give
\begin{equation}
  \kappa_C(q)=
    \frac{P_0(x_s+q)(x_s+(1-f)q)}
         {W_0[(x_s+(1-f)q)+(1-f)(x_s+q)]}.
  \label{eq:exp2-cpmm-capacity}
\end{equation}
Local normalized price sensitivity is $1/\kappa$, with zero capacity
corresponding to positive-infinite sensitivity.  A larger value means that
the same small addition to normalized input produces a larger log-price
change.  Together with average execution and the terminal price, these
measures describe how the pool responds throughout settlement.

\paragraph{Use of generative AI}
In accordance with ACM's authorship policy, we describe the role of
generative AI in this work.
The lead author developed the core idea of homothetic liquidity scaling and the
original asymptotic analysis, and designed the experimental framework linking
annual calibration in Experiment~1 to the acute-order evaluation in
Experiment~2 through the same active fraction $N^*$.  Across the mathematical
analysis and both experiments, OpenAI Codex assisted under the lead author's
direction with checking and refining mathematical
derivations, execution and buffer-recovery rules, and experimental procedures,
including sampling, historical-state reconstruction, and the matched-order
comparison.  To support reproducible evaluation, Codex assisted with
developing and running the computational workflow for API-based data
acquisition, preprocessing, simulation, and numerical analysis, together with
automated verification tests.  It also assisted with figure preparation and
manuscript drafting and revision.  The lead author made the final
methodological and interpretive decisions, and the authors retain
responsibility for the work.

\section{Empirical Validation and Discussion}
\label{sec:evaluation}

\Cref{sec:method} defines two linked experiments: an annual calibration of
the initial active fraction and an acute-order comparison that retains
this fraction when evaluating buffer activation.  We now examine what
their results show about the execution cost of keeping capital outside
the active curve and the depth restored when that capital is deployed.
The annual result establishes the reserve split carried into the acute
case, so the two evaluations follow the same capital allocation from
initial depth to activated depth.  \Cref{sec:results-exp1} relates the
calibrated buffer share to its additional output loss relative to full
depth.  \Cref{sec:results-exp2} examines how the common order executes in
the three pools and what activation changes for the trader and the pool's
price path.  \Cref{sec:interpretation} then connects these findings to
HLCP's reserve-allocation mechanism and the evidence supporting that
interpretation.

\subsection{Annual capital-scale calibration}
\label{sec:results-exp1}

Experiment~1 evaluates how much capital can remain in the paired buffer
while limiting the additional output loss relative to a full-depth CPMM.
For the normalized historical orders described in \cref{sec:method-exp1},
the annual calibration leaves 18.5\% of balanced synthetic capital in the
paired buffer while keeping the additional output loss relative to full
depth within 5~bp (0.05\%) for at least 95\% of the eligible sampled Swaps.
The corresponding active fraction is $N^*=0.815$.  This result follows from
the observed order-to-capital ratios: among the 283 eligible Swaps, the
median normalized order size was $6.13\times10^{-5}$ and the Type-1 P95 was
\begin{equation}
  z_{.95}=1.10\times10^{-3}.
  \label{eq:observed-z95}
\end{equation}
This is an input value of 0.110\% of the corresponding pre-swap marked-capital
proxy.  It has rank 269 of 283, with equal weight for each eligible Swap.
Substituting the unrounded percentile into \cref{eq:p95-gap} gives the curve
in \cref{fig:exp1-calibration}.

\begin{figure}[htbp]
  \centering
\begin{tikzpicture}
  \begin{axis}[
    width=0.84\textwidth,height=5.7cm,
    xlabel={Capital held in the buffer (\%)},
    ylabel={P95 output shortfall (bp)},
    xmin=0,xmax=50,ymin=0,ymax=24,
    xtick={0,10,20,30,40,50},
    ytick={0,5,10,15,20},
    grid=major,
    major grid style={draw=black!12,line width=0.3pt},
    tick label style={font=\scriptsize},
    label style={font=\small},
    legend style={draw=none,fill=none,font=\scriptsize,
                  at={(0.025,0.975)},anchor=north west},
    clip=true
  ]
    \addplot[hlcpblue,very thick,domain=0:50,samples=101]
      {10000*(2*0.9995*0.001104663259652063)*(x/100)
       /((1-x/100)+2*0.9995*0.001104663259652063)};
    \addlegendentry{Eligible-sample P95}
    \addplot[hlcpgray,dashed,semithick] coordinates {(0,5) (50,5)};
    \addlegendentry{5-bp fee scale}
    \draw[hlcpgray,densely dotted]
      (axis cs:18.5030672361501427,0) -- (axis cs:18.5030672361501427,5);
    \fill[hlcporange] (axis cs:18.5030672361501427,5) circle[radius=2pt];
    \draw[hlcpgray,thin]
      (axis cs:18.5030672361501427,5) -- (axis cs:22,3.3);
    \node[anchor=west,align=left,font=\scriptsize,inner sep=1pt]
      at (axis cs:22,2.7) {18.5\% buffer\\$N^*=0.815$};
  \end{axis}
\end{tikzpicture}
  \caption{Buffer allocation and swap execution in Experiment~1.  The curve
  evaluates \cref{eq:p95-gap} at the eligible-sample P95 order size, showing
  output shortfall relative to the full-depth synthetic CPMM.  Its 5-bp
  crossing leaves 18.5\% of capital in the paired buffer.}
  \Description{An increasing curve relates buffer capital from zero to fifty
  percent to the P95 output shortfall. A horizontal five-basis-point fee
  reference intersects the curve at an 18.5 percent buffer share, corresponding
  to an active fraction of 0.815.}
  \label{fig:exp1-calibration}
\end{figure}

The curve makes the capital-allocation trade-off explicit.  A 10.0\% buffer
gives a 2.45-bp P95 gap, while a 25.0\% buffer raises it to 7.34~bp.
The 18.5\% crossing provides the intermediate reserve split used in
Experiment~2: 81.5\% begins on the active curve, and the remaining paired
capital is available for activation when the order requires more depth.

\FloatBarrier
\subsection{Acute-tail execution paths}
\label{sec:results-exp2}

Experiment~1 identifies the initial reserve split carried into the acute
case.  Keeping this split unchanged, we now examine how deploying the
existing buffer changes execution of the common large order.  We first
describe the order and the depth restored before settlement, then examine
the resulting average execution price, final pool price, and local execution
capacity.  The matched static CPMM isolates the effect of activation, while
historical V3 places the result alongside the response of one observed
passive-liquidity allocation.

The selected event was a 9.20\% ETH--USDC price decline on February 3,
2025, from 01:45 to 02:00 UTC.  The pre-event V3 snapshot had
$P_0=2.75\times10^3$ USDC/WETH and $W_0=161$ million USDC under the
marked-capital proxy.  Reaching the mapped endpoint in V3 required a gross
order of $Q=1.12\times10^3$ WETH, equal to 1.91\% of $W_0$ at $P_0$.
All three pools filled this same order completely; V3 crossed 92 initialized
ticks without a zero-liquidity gap.  The comparison therefore concerns the
prices received during execution and the pool state left by the order,
rather than differences in the amount filled.

For the matched synthetic pools, the difference begins with the depth
available when settlement starts.  HLCP's pre-swap preview projected a
price move beyond the controller's tolerance, with the excess measured
by $\phi_t=0.0377$.
The controller therefore deployed 75.4\% of the paired buffer capacity
before the order executed.  This deposit raised the depth at the start of
settlement from 81.5\% to 95.5\% of the balanced full-depth benchmark while
preserving the pre-settlement price.  The static CPMM kept its paired buffer
dormant, so the common order began against the original reduced reserves
in that arm and the expanded reserves in HLCP.  The two pools began with
the same active reserves and total marked capital; activation changed how
much of that existing capital participated in settlement.

\begin{table}[t]
  \caption{Experiment~2 same-order results, rounded to three significant
  digits.  Log displacement is terminal $p(u)$, and VWAP deterioration is the
  relative execution loss in \cref{eq:down-vwap}.  All four metrics are
  dimensionless.  Smaller displacement and deterioration indicate less
  order-induced price impact and better average execution; lower sensitivity
  and greater depth indicate more local execution capacity.}
  \label{tab:exp2-results}
  \centering
  \small
  \begin{tabular}{lrrrr}
    \toprule
    Arm & Log displ. & VWAP deter. & Worst sens. & Min depth \\
    \midrule
    Passive V3      & 0.0965 & 0.0406 & 9.88 & 0.368 \\
    Static $N^*$    & 0.0918 & 0.0453 & 4.91 & 0.778 \\
    HLCP $N^*$      & \textbf{0.0786} & \textbf{0.0390}
                    & \textbf{4.19} & \textbf{0.918} \\
    \bottomrule
  \end{tabular}
\end{table}

This restored depth improved the trader's execution in the matched
comparison.  For the same fully filled WETH order, HLCP delivered 0.660\%
more USDC than the static CPMM.  The additional USDC received was equivalent
to 12.6 times the pool's input fee, valued at the pre-order price $P_0$.
After the order was fully executed, the pool's marginal price was 7.56\%
below its pre-order level $P_0$, compared with 8.77\% in the static CPMM.
The trader thus received more USDC while leaving the pool price closer to
its starting level.  Expressed in the complete-order metrics of
\cref{tab:exp2-results}, these outcomes correspond to a 14.3\% reduction in
terminal log displacement and a 13.9\% reduction in gross VWAP
deterioration.

\Cref{fig:exp2-paths} shows how the pool's price and the trader's average
execution develop as the common order is filled.  Panel~(a) records the
marginal-price decline, while panel~(b) records cumulative average
execution-price deterioration, both relative to $P_0$.  Reading the panels
together reveals why the final pool price alone cannot determine execution
quality: the static CPMM ends with a smaller marginal-price decline than V3,
yet gives the trader a worse average execution price.  The endpoint describes
the price left in the pool, whereas average execution reflects the prices
received over the whole order.  In the selected case, HLCP improves both
outcomes relative to the two comparators.

\begin{figure*}[t]
  \centering
  \begin{tikzpicture}
    \begin{axis}[
      width=0.94\textwidth,height=5.3cm,
      title={(a) Price response to the common order},
      title style={font=\small\bfseries},
      xlabel={Gross WETH input valued at $P_0$ / $W_0$ (\%)},
      ylabel={Price decline from $P_0$ (\%)},
      xmin=0,xmax=2.0,ymin=0,ymax=10.3,
      xtick={0,0.5,1.0,1.5,2.0},
      ytick={0,2,4,6,8,10},
      grid=major,
      major grid style={draw=black!12,line width=0.3pt},
      tick label style={font=\scriptsize},
      label style={font=\small},
      legend columns=-1,
      legend style={draw=none,fill=none,font=\scriptsize,
                    at={(0.02,0.98)},anchor=north west,
                    /tikz/every even column/.append style={column sep=5pt}},
      clip=false
    ]
      \addplot[hlcpgray,semithick]
        table[x expr=100*\thisrow{u},
              y expr=100*(1-exp(-\thisrow{p})),col sep=comma]
        {figures/data/PASSIVE_ACTUAL_V3_PATH.csv};
      \addlegendentry{Passive V3}
      \addplot[hlcporange,thick,dashed]
        table[x expr=100*\thisrow{u},
              y expr=100*(1-exp(-\thisrow{p})),col sep=comma]
        {figures/data/STATIC_N_STAR_CPMM_PATH.csv};
      \addlegendentry{Static CPMM}
      \addplot[hlcpblue,very thick]
        table[x expr=100*\thisrow{u},
              y expr=100*(1-exp(-\thisrow{p})),col sep=comma]
        {figures/data/HLCP_N_STAR_PATH.csv};
      \addlegendentry{HLCP}

      \fill[hlcpgray] (axis cs:1.91347,9.19822) circle[radius=1.3pt];
      \node[anchor=east,font=\scriptsize,text=hlcpgray,inner sep=1pt]
        at (axis cs:1.98,9.70) {9.20\%};
      \fill[hlcporange] (axis cs:1.91347,8.76715) circle[radius=1.3pt];
      \node[anchor=east,font=\scriptsize,text=hlcporange,inner sep=1pt]
        at (axis cs:1.98,8.05) {8.77\%};
      \fill[hlcpblue] (axis cs:1.91347,7.55894) circle[radius=1.3pt];
      \node[anchor=east,font=\scriptsize,text=hlcpblue,inner sep=1pt]
        at (axis cs:1.98,7.05) {7.56\%};
    \end{axis}
  \end{tikzpicture}

  \vspace{0.10em}

  \begin{tikzpicture}
    \begin{axis}[
      width=0.94\textwidth,height=5.3cm,
      title={(b) Cumulative average execution-price deterioration},
      title style={font=\small\bfseries},
      xlabel={Gross WETH input valued at $P_0$ / $W_0$ (\%)},
      ylabel={Average price deterioration (\%)},
      xmin=0,xmax=2.0,ymin=0,ymax=5.2,
      xtick={0,0.5,1.0,1.5,2.0},
      ytick={0,1,2,3,4,5},
      grid=major,
      major grid style={draw=black!12,line width=0.3pt},
      tick label style={font=\scriptsize},
      label style={font=\small},
      legend columns=-1,
      legend style={draw=none,fill=none,font=\scriptsize,
                    at={(0.02,0.98)},anchor=north west,
                    /tikz/every even column/.append style={column sep=5pt}},
      clip=false
    ]
      \addplot[hlcpgray,semithick]
        table[x=input_pct,y=deterioration_pct,col sep=comma]
        {figures/data/PASSIVE_ACTUAL_V3_VWAP.csv};
      \addlegendentry{Passive V3}
      \addplot[hlcporange,thick,dashed]
        table[x=input_pct,y=deterioration_pct,col sep=comma]
        {figures/data/STATIC_N_STAR_CPMM_VWAP.csv};
      \addlegendentry{Static CPMM}
      \addplot[hlcpblue,very thick]
        table[x=input_pct,y=deterioration_pct,col sep=comma]
        {figures/data/HLCP_N_STAR_VWAP.csv};
      \addlegendentry{HLCP}

      \fill[hlcporange] (axis cs:1.91347,4.53080) circle[radius=1.3pt];
      \node[anchor=east,font=\scriptsize,text=hlcporange,inner sep=1pt]
        at (axis cs:1.98,4.80) {4.53\%};
      \fill[hlcpgray] (axis cs:1.91347,4.06104) circle[radius=1.3pt];
      \node[anchor=east,font=\scriptsize,text=hlcpgray,inner sep=1pt]
        at (axis cs:1.98,4.28) {4.06\%};
      \fill[hlcpblue] (axis cs:1.91347,3.90087) circle[radius=1.3pt];
      \node[anchor=east,font=\scriptsize,text=hlcpblue,inner sep=1pt]
        at (axis cs:1.98,3.60) {3.90\%};
    \end{axis}
  \end{tikzpicture}
  \caption{Execution of the common order calibrated to reach the V3 target
  endpoint.  (a) Marginal price decline and (b) cumulative average
  execution-price deterioration, including fees, are measured relative to
  $P_0$.  Both horizontal axes show cumulative gross WETH input valued at
  $P_0$ as a percentage of the initial marked-capital proxy $W_0$.
  HLCP begins settlement after its buffer deposit; the static CPMM keeps
  its buffer dormant.}
  \Description{Two vertically stacked panels compare three mechanisms over
  the same order. The upper panel shows percentage price decline from the
  common pre-order price, ending at 9.20 percent for passive V3, 8.77 percent
  for static CPMM, and 7.56 percent for HLCP. The lower panel shows cumulative
  average execution-price deterioration, including fees, ending at 4.06
  percent for passive V3, 4.53 percent for static CPMM, and 3.90 percent for
  HLCP. HLCP begins settlement after its full-order buffer deposit.}
  \label{fig:exp2-paths}
\end{figure*}

Having established the activation effect against the matched static CPMM,
we place the resulting execution alongside the historical V3 allocation.
In this snapshot, deploying the existing buffer reversed the static pool's
average-execution disadvantage: HLCP delivered 0.167\% more USDC than V3 for
the same fully filled order.  The additional output was equivalent to 3.20
times the pool's input fee, valued at $P_0$.  The trader received more USDC
while the pool's final price remained closer to its starting level.

These complete-order outcomes describe the trader's proceeds and the final
pool price.  The local metrics in \cref{tab:exp2-results} complement them by
measuring the least depth and greatest sensitivity encountered during
execution.  Relative to the static $N^*$ arm, HLCP's minimum local depth
increased by 17.9\% and its worst local sensitivity fell by 14.6\%; relative
to historical V3, it likewise had greater minimum depth and lower worst
sensitivity.  The improvement in proceeds was therefore accompanied by
greater minimum execution capacity and a smaller maximum price response
to a small input increment.  These gains came from deploying part of the
existing paired buffer before
settlement, allowing capital held outside the initial trading curve to
support the order when greater depth was needed.

\subsection{Interpretation and scope}
\label{sec:interpretation}

HLCP's reserve split creates an execution trade-off: keeping paired capital
in the buffer reduces the depth initially available to swaps, while
activation allows that capital to participate when a large order needs it.
The two experiments evaluate these parts of the design in sequence.
Experiment~1 compares reduced and full active reserves for normalized
historical order sizes, identifying a split at which the eligible-sample
P95 output gap equals the selected pool's fee scale.  Experiment~2 holds
that split fixed and evaluates the execution of a common large order with
and without buffer activation.  The resulting improvement connects the
cost of leaving capital outside the initial trading curve to the depth
that the same capital restores in the acute case.

The restored depth comes from tokens already included in the pool's
committed capital.  Their deposit changes when that capital participates
in the swap and expands the active curve homothetically without changing
its starting price.  HLCP and the static $N^*$ CPMM share the same balanced
reserve geometry, initial capital, fee, and order, so their comparison
identifies the execution effect of activation.  Deploying part of the
existing buffer improves the trader's average execution and leaves the
pool's final price closer to its starting level than the matched static
pool.  The historical V3 arm retains its observed curve geometry and
distribution of liquidity across price ranges, placing this effect
alongside the response of one passive allocation.  With complete fills in
all three arms, the differences concern the prices received and the local
capacity available as the order progresses.

This execution result returns to the economic motivation for the reserve
split.  The capital that supplies depth on the active curve is also
available to arbitrageurs when external prices change before the pool
adjusts.  In the LVR framework, external-price variation is weighted by the
liquidity present during repricing \cite{LVR,Modeling_LVR}.  HLCP holds a
paired share outside that curve and temporarily deploys it when the
submitted order needs depth.  The annual calibration measures the
execution cost of this initial allocation, and the acute case shows how
its buffer can improve execution when activated.  These results explain
the execution side of reducing continuously active reserves.  Quantifying
cumulative LVR or LP return would additionally require an external price
path, an arbitrage-clearing rule, and the corresponding LP benchmark.

The interpretation follows the observations used in the two comparisons.
The annual calibration describes reconstructable sampled Swaps from one
WETH/USDC pool in 2025, while the acute comparison describes one downward
event at its historical snapshot.  Both use the pool's token balances
after subtracting accrued protocol fees to measure marked capital and
reconstitute balanced reserves for the synthetic CPMMs.  These balances
may contain donations or other non-LP surplus; such surplus increases the
capital proxy and can bias the normalized order sizes and $N^*$ downward.
The evidence therefore supports the reserve-allocation comparison at
this measured capital scale and the activation effect in the selected
case.

\section{Conclusion}
\label{sec:conclusion}

This paper studies how HLCP divides fully backed reserves between
continuously active liquidity and a paired buffer that restores depth for
a stressed swap.  The secant and endpoint-tangent relations explain how
active liquidity affects the trader's average execution price and the
pool's final marginal price.  HLCP uses a preview of the submitted order to
set a bounded buffer deposit, expanding the active curve homothetically
without changing its starting price.  The same order then executes on the
deeper curve, and the activated share returns to the buffer in its
post-trade token composition.  A simple reserve transition therefore makes
execution depth responsive to the submitted order while preserving physical
token backing.

The annual calibration identifies a reserve split that leaves 18.5\% of
balanced synthetic capital in the buffer, with an initial active fraction
$N^*=0.815$.  At this split, the eligible-sample P95 output gap equals the
selected pool's 5-bp fee scale.  Holding that initial split fixed, the
acute case shows that deploying part of the buffer improves average
execution and limits the final marginal-price decline relative to the
matched static CPMM.  It also improves the minimum local depth and worst
local price sensitivity, with historical V3 providing the observed
passive-liquidity comparison.  Together, the experiments connect the
execution cost of leaving capital dormant to the depth restored by buffer
activation for a large order.

\bibliographystyle{ACM-Reference-Format}
\bibliography{hlcp_sigmetrics_2027}

\appendix
\section{Mechanism Ordering and Token Accounting}
\label{app:mechanism-accounting}

This appendix records the token accounting implied by the HLCP transition in
\cref{sec:architecture}.  One submitted order defines one atomic event with
the following order of operations.

\begin{enumerate}
  \item \textbf{Read the pre-state.}  Snapshot active reserves
  $(x_t,y_t)$ and buffer balances $(B^x_t,B^y_t)$; compute $P_t$, $L_t$,
  and the paired capacity $C_t$ from \cref{eq:paired-capacity}.

  \item \textbf{Preview without mutation.}  Apply the fee-aware
  counterfactual transition $T_{f_t}$ to the submitted order $o_t$ and compute
  the prospective stress $\phi_t$ in \cref{eq:prospective-stress}.  The
  preview is not a settlement and changes no balance.

  \item \textbf{Activate once.}  Compute $c_t=\Delta C_t$ from the sole
  controller in \cref{eq:single-controller} and transfer the vector
  \[
    \left(c_t/\sqrt{P_t},\,c_t\sqrt{P_t}\right)
  \]
  from the buffer to the active pool.  This deposit is fully token-backed and
  preserves the pre-settlement price.

  \item \textbf{Settle once.}  Execute the same submitted order exactly once
  after the homothetic expansion.  Denote the resulting fee-retaining active
  reserves before recovery by $(x_t^{\mathrm{tr}},y_t^{\mathrm{tr}})$.  The
  realized price is not used to recompute activation within this event.

  \item \textbf{Recover and commit.}  Remove the conceptual LP share minted
  by the temporary deposit and return its actual post-trade token composition
  to the buffer.  The resulting active and buffer balances become the next
  event's pre-state.
\end{enumerate}

For $c_t>0$, the temporary share is
\begin{equation}
  \omega_t=\frac{c_t}{L_t+c_t}.
  \label{eq:app-activated-share}
\end{equation}
Writing the next event's pre-state with subscript $t+1$, the exact
active-reserve update is
\begin{align}
  x_{t+1}&=(1-\omega_t)x_t^{\mathrm{tr}},
  &
  y_{t+1}&=(1-\omega_t)y_t^{\mathrm{tr}},
  \label{eq:app-active-recovery}
\end{align}
and the buffer update is
\begin{align}
  B^x_{t+1}
    &=B^x_t-\frac{c_t}{\sqrt{P_t}}
      +\omega_t x_t^{\mathrm{tr}},
  &
  B^y_{t+1}
    &=B^y_t-c_t\sqrt{P_t}
      +\omega_t y_t^{\mathrm{tr}}.
  \label{eq:app-buffer-recovery}
\end{align}
Thus the deposit and withdrawal cancel internally in each token account; the
only change in aggregate active-plus-buffer balances is the external swap
itself.  When $c_t=0$, no temporary share is created and the order settles in
the active pool without a preceding buffer deposit.

Because withdrawal is pro rata from post-trade reserves, the recovered vector
need not lie on the deposit price ray; the next event therefore recomputes
$C_{t+1}$ from the realized buffer balances.  Settlement exposure includes the
temporary deposit, whereas end-of-event continuous exposure follows
\cref{eq:app-active-recovery,eq:app-buffer-recovery}.  Under the assumptions of
Proposition~\ref{prop:hlcp-core}, the bound
$0\leq c_t<C_t$ is event-local and does not ensure buffer sufficiency across
orders.

\end{document}